\documentclass[11pt]{article}
\usepackage[margin=1in]{geometry}
\usepackage{complexity}
\usepackage{mathtools,amsmath,amssymb}
\usepackage{xspace}
\usepackage{xcolor}
\RequirePackage[colorlinks=true]{hyperref}
\hypersetup{
	linkcolor=[rgb]{0,0,0.5},
	citecolor=[rgb]{0, 0.5, 0},
	urlcolor=[rgb]{0.5, 0, 0}
}
\usepackage{dsfont}

\usepackage{amsthm}
\usepackage{thmtools,thm-restate}

\numberwithin{equation}{section}
\declaretheoremstyle[bodyfont=\it,qed=\qedsymbol]{noproofstyle}

\declaretheorem[numberlike=equation]{observation}

\declaretheorem[name=Observation,numbered=no]{observation*}

\declaretheorem[numberlike=equation]{theorem}

\declaretheorem[name=Theorem,numbered=no]{theorem*}

\declaretheorem[numberlike=equation]{lemma}
\declaretheorem[name=Lemma,numbered=no]{lemma*}

\declaretheorem[numberlike=equation]{corollary}
\declaretheorem[name=Corollary,numbered=no]{corollary*}

\declaretheorem[name=Proposition,numbered=no]{proposition*}

\declaretheorem[name=Claim,numbered=no]{claim*}

\declaretheorem[name=Conjecture,numbered=no]{conjecture*}

\declaretheorem[name=Question,numbered=no]{question*}

\declaretheoremstyle[bodyfont=\it,qed=$\lozenge$]{defstyle} 

\declaretheorem[numberlike=equation,style=defstyle]{definition}
\declaretheorem[unnumbered,name=Definition,style=defstyle]{definition*}

\declaretheorem[unnumbered,name=Example,style=defstyle]{example*}

\declaretheorem[unnumbered,name=Notation=defstyle]{notation*}

\declaretheorem[unnumbered,name=Construction,style=defstyle]{construction*}

\declaretheorem[numberlike=equation,style=defstyle]{remark}
\declaretheorem[unnumbered,name=Remark,style=defstyle]{remark*}

\usepackage{nth}
\usepackage{intcalc}
\usepackage{etoolbox}
\usepackage{xstring}
\hypersetup{
}

\usepackage{ifpdf}
\ifpdf
\else
\usepackage[quadpoints=false]{hypdvips}
\fi

\newcommand{\ECCCurl}[2]{http://eccc.weizmann.ac.il/report/\ifnumcomp{#1}{>}{93}{19}{20}#1/#2/}

\newcommand{\shortECCC}[2]{\texttt{\href{http://eccc.weizmann.ac.il/report/\ifnumcomp{#1}{>}{93}{19}{20}#1/#2/}{eccc:TR#1-#2}}}

\newcommand{\parseECCC}[1]{% Takes a string of the form TRxx/xxx or
\StrSubstitute{#1}{TR}{}[\tmpstring]%
\IfSubStr{\tmpstring}{/}{ %assuming string is of the form TRxx/xxx
\StrBefore{\tmpstring}{/}[\ecccyear]%
\StrBehind{\tmpstring}{/}[\ecccreport]%
}{% assuming string is of the form TRxx-xxx
\StrBefore{\tmpstring}{-}[\ecccyear]%
\StrBehind{\tmpstring}{-}[\ecccreport]%
}%
\shortECCC{\ecccyear}{\ecccreport}}

\renewcommand{\vec}[1]{{\mathbf{#1}}}

\makeatletter
\newcommand{\va}{{\vec{a}}\@ifnextchar{^}{\!\:}{}}
\newcommand{\vb}{{\vec{b}}\@ifnextchar{^}{\!\:}{}}
\newcommand{\vc}{{\vec{c}}\@ifnextchar{^}{\!\:}{}}
\newcommand{\vd}{{\vec{d}}\@ifnextchar{^}{\!\:}{}}
\newcommand{\ve}{{\vec{e}}\@ifnextchar{^}{\!\:}{}}
\newcommand{\vy}{{\vec{y}}\@ifnextchar{^}{\!\:}{}}
\newcommand{\vs}{{\vec{s}}\@ifnextchar{^}{\!\:}{}}
\newcommand{\vt}{{\vec{t}}\@ifnextchar{^}{\!\:}{}}
\newcommand{\vx}{{\vec{x}}\@ifnextchar{^}{}{}}		%\vec{x} seems fine already
\newcommand{\vz}{{\vec{z}}\@ifnextchar{^}{\!\:}{}}
\newcommand{\vw}{{\vec{w}}\@ifnextchar{^}{}{}}		%\vec{x} seems fine already
\newcommand{\vh}{{\vec{h}}\@ifnextchar{^}{\!\:}{}}
\newcommand{\vg}{{\vec{g}}\@ifnextchar{^}{\!\:}{}}
\newcommand{\vv}{{\vec{v}}\@ifnextchar{^}{\!\:}{}}
\newcommand{\vp}{{\vec{p}}\@ifnextchar{^}{\!\:}{}}
\newcommand{\vq}{{\vec{q}}\@ifnextchar{^}{\!\:}{}}

\newcommand{\vY}{{\vec{Y}}\@ifnextchar{^}{\!\:}{}}
\newcommand{\vX}{{\vec{X}}\@ifnextchar{^}{}{}}		%\vec{x} seems fine already
\newcommand{\vZ}{{\vec{Z}}\@ifnextchar{^}{\!\:}{}}
\newcommand{\vG}{{\vec{G}}\@ifnextchar{^}{\!\:}{}}

\makeatother

\newcommand{\F}{\mathbb{F}}

\newcommand{\Q}{\mathbb{Q}}
\newcommand{\N}{\mathbb{N}}
\newcommand{\Z}{\mathbb{Z}}

\newcommand{\set}[1]{\left\{#1\right\}}

\def\epsilon{\varepsilon}

\newcommand{\Cnote}[1]{\begin{quote}{\textcolor{violet}{{\sf Chi-Ning's Note:} {#1}}}\end{quote}}

\newcommand{\vu}{\mathbf{u}}

\date{}
\title{Closure of $\VP$ under taking factors: a short and simple proof}
\author{
	Chi-Ning Chou\thanks{School of Engineering and Applied Sciences, Harvard University, Cambridge, Massachusetts, USA. Supported by Boaz Barak's NSF awards CCF 1565264 and CNS 1618026. Email: \texttt{chiningchou@g.harvard.edu}.}
	\and
	Mrinal Kumar\thanks{Department of Computer Science, University of Toronto, Canada. A part of this work was done during the semester on Lower bounds in Computational Complexity at  Simons Institute for the Theory of Computing, Berkeley, CA, USA. Email: \texttt{mrinalkumar08@gmail.com}.}
	\and
	Noam Solomon\thanks{Department of Mathematics, MIT, Cambridge, Massachusetts, USA. Email: \texttt{noam.solom@gmail.com}.}
}
\begin{document}
	
\maketitle

\begin{abstract}
In this note, we give a short, simple and almost completely self contained proof of a classical result of Kaltofen~\cite{Kal86, Kal87, kalt89} which shows that if an $n$ variate degree $d$ polynomial $f$ can be computed by an arithmetic circuit of size $s$, then each of its factors can be computed by an arithmetic circuit of size at most $\poly\left(s, n, d\right)$. 

However, unlike Kaltofen's argument, our proof  does not  directly give an efficient algorithm for computing the circuits for the factors of $f$. 
\end{abstract}

\thispagestyle{empty}
%\newpage
\pagenumbering{arabic}
\iffalse
\section*{Things to do}
\begin{itemize}
\item Finishing up the proof of the pure power case. 
\item Adding in preliminaries and a short intro. 
\item Some further clean up is needed in the notations
\item we should perhaps compare the proof with the two main proofs we know, Kaltofen's original proof, and Burgisser's argument via univariate polynomials. 
\end{itemize}
\fi
\section{Introduction}
\textit{Polynomial factorization} is a fundamental problem at the intersection of algebra and computation and has been intensively studied in algebraic complexity theory. Given a multivariate polynomial $f\in\F[x_1,x_2,\dots,x_n]$, the goal is to find  irreducible factors of $f$. The nature of these algorithms as well as their efficiency varies depending upon on how the input polynomial is given. Two natural representations which are often used in this context are the monomial representation (where the polynomial is given as sum of its monomials) and the circuit representation (where the polynomial is given as an arithmetic circuit). In this note, we focus on the latter. In this setting, we are given an arithmetic circuit computing a multivariate polynomial, and the goal is to output arithmetic circuits for all its irreducible factors. The problem has been studied in both the whitebox setting (where we have access to the internal wirings of the input circuit) and in the blackbox setting (where we only have query access to the input circuit). In a sequence of extremely influential results in the 1980's~\cite{Kal86, Kal87, kalt89, KT90},  Kaltofen (and Kaltofen and Trager~\cite{KT90}) gave efficient randomized algorithms for this problem. A consequence of these results which has had extremely interesting applications in algebraic complexity theory~\cite{ki03} is that if an $n$-variate degree $d$ polynomial has an arithmetic circuit of size $s$, then each of its factors has an arithmetic circuit of size $\poly(s, n, d)$. In other words, the complexity class $\VP$ of polynomials is closed under taking factors. 

In addition to being natural mathematical questions on their own, these \emph{closure results} for polynomial factorization seem crucial to our current understanding of hardness randomness tradeoffs in algebraic complexity~\cite{ki03, DSY09, CKS18A}. In this note, we give a short, simple and almost completely self contained proof of the closure of $\VP$ under taking factors. More formally, we give a new\footnote{as far as we know.} proof of the following result of Kaltofen. 
%Since the 1960's, efficient algorithms for polynomial factorization as well as many applications in algebraic complexity, coding theory, and cryptography have been discovered. The com
%For more background and applications of polynomial factorization, we refer the interested reader to a recent survey of Forbes and Shpilka~\cite{FS15}.

%In this note, we focus on the following complexity question about polynomial factorization.
%\begin{center}
%Suppose $f$ is in the complexity class $\class{C}$, does its factor(s) also lie in $\class{C}$?
%\end{center}
%The seminal work by Kaltofen~\cite{kalt89} in the 1980's showed that this is the case for $\class{C}=\VP$, the class of polynomials that have polynomial degree and can be computed by polynomial size arithmetic circuits.
\begin{theorem}[Kaltofen]\label{thm:main}
Let $f \in \F[x_1,x_2,\dots,x_{n-1}, y]$ be an $n$-variate degree $d$ polynomial which can be computed by an arithmetic circuit of size $s$. Let $g$ be a polynomial such that $g$ divides $f$. Then, $g$ can be computed by an arithmetic circuit of size at most $\poly(s,n,d)$. 
\end{theorem}
The original proof of~\autoref{thm:main} relies on some beautiful and neat mathematical ideas like Hensel's lifting, effective Hilbert's Irreducibility Theorem, etc, which are useful and interesting on their own.  For our proof, we only rely on a simple and natural multivariate version of the classical Newton Iteration technique and the fact that the Resultant of two univariates tells us exactly when they have a non-trivial greatest common divisor $(\mathsf{GCD})$. We hope that this simpler proof can shed some more insight on this closure result (and hopefully some others, which are yet to be discovered), and is more accessible to readers with a less detailed background in algebra. A cost of this \emph{simplicity} is that unlike in the work of Kaltofen, we do not get an algorithm for factoring multivariate polynomials given by arithmetic circuits. 

Besides Kaltofen's original proof, there is a considerably simpler proof due to B\"{u}rgisser~\cite{bur04} showing that $\VP$ is closed under taking factors. B\"{u}rgisser uses the classical univariate Newton Iteration to  obtain a power series approximation of a root of a multivariate polynomial when it is viewed as a univariate in one of the variables. This power series approximation of the root to a sufficiently high enough accuracy is then used to obtain an appropriate irreducible factor of the input polynomial. This step requires setting up and solving an appropriate system of linear equations. A variant of this argument is also present in the works of Dvir et al.~\cite{DSY09}, of Oliveira~\cite{O16}, of Dutta et al.~\cite{DSS17} and an earlier work of the authors~\cite{CKS18A, CKS18B}. At a high level, these proofs go via an iterative step to approximate a root (or many roots), and a clean up step where a factor is recovered from this approximation. 

In our proof, we directly recover the factors at the end of the slightly more complicated iterative step (a \emph{multivariate} Newton iteration as opposed to a \emph{univariate} one), and the clean up is essentially trivial.  

Our proof follows immediately from the following two lemmas. 
\begin{lemma}\label{clm:distinct factors}
Let $f \in \F[x_1, x_2, \ldots, x_{n-1}, y]$ be an $n$-variate degree $d$ polynomial which can be computed by an arithmetic circuit of size at most $s$.  If $g$ and $h$ are polynomials of degree at least $1$ such that $f = g\cdot h$ and $\mathsf{GCD}(g,h) = 1$, then $g$ and $h$ have a circuit of size at most $\poly(s, n, d)$. 
\end{lemma}

\begin{lemma}\label{clm:purepower}
Let $f \in \F[x_1, x_2, \ldots, x_{n-1}, y]$ be an $n$-variate degree $d$ polynomial which can be computed by an arithmetic circuit of size at most $s$.  If there is a polynomial $g$ and an integer $e$ such that $f = g^e$, then $g$ has a circuit of size at most $\poly(s, n, d)$. 
\end{lemma}

The proof of~\autoref{clm:distinct factors} and~\autoref{clm:purepower} will be provided in~\autoref{sec:proof claim 1} and~\autoref{sec:proof claim 2} respectively. 

\iffalse
Before that, we are going to briefly mention some related works in~\autoref{sec:related} and provide some backgrounds in~\autoref{sec:prelims}.

\subsection{Related works}\label{sec:related}

\paragraph{Previous proofs for~\autoref{thm:main}}
The idea of Kaltofen's original proof for~\autoref{thm:main} consists of two steps. First, reducing a multivariate polynomial to a \textit{bivariate polynomial}. Usually this step can be done efficiently with the help of randomness. The second step factors a bivariate polynomial using the idea of Hensel's lifting (or Newton's iteration).

\Cnote{Maybe more details on the previous proofs?}

\paragraph{Closure results for other complexity classes}
B\"{u}rgisser~\cite{bur04} showed that the class of polynomials that can be arbitrarily approximated by polynomial size arithmetic circuits\footnote{Note that there's no degree restriction here.} is closed under factoring. Dutta, Saxena, and Sinhababu~\cite{DSS18} showed that quasi-polynomial size arithmetic formulas and arithmetic branching programs  are closed under factoring. The authors of this notes~\cite{CKS18A} showed that $\VNP$ is closed under factoring. Dvir, Shpilka, and Yehudayoff~\cite{DSY09}, Oliveira~\cite{O16}, and the authors of this notes~\cite{CKS18B} also showed closure results for bounded depth arithmetic circuits with certain degree constraints.

On the other hand, the class of sparse polynomials does not have the closure property. Specifically, there exists polynomial of whose factor having a quasi-polynomial blow-up in the sparsity~\cite{vzGK85}.
\fi
\section{Notations and Preliminaries}\label{sec:prelims}
We follow the following notation. 
\begin{itemize}
\item Throughout the paper, $\F$ is a field of characteristic zero or sufficiently large.
\item For a positive integers $n$, $[n]$ denotes the set $\{0, 1, 2, \ldots, n-1\}$.
\item We use boldface letters to denote ordered tuples of objects. For instance, $\vx = (x_1, x_2, \ldots, x_n)$, or $\vg = (g_0, g_1, \ldots, g_{d_1-1})$. The length of these tuples and the precise indexing is defined before the specific notation is invoked. The sum of two such tuples of the same length is their coordinate wise sum.
\item We say that a function $\Psi$ of parameters $n_1, n_2, \ldots, n_t$ taking values in $\Z^{+}$ is $\poly(n_1, n_2, \ldots, n_t)$ if there is a polynomial $\Phi$ such that for all sufficiently large values of $n_1, n_2, \ldots, n_t$, $\Psi(n_1, n_2, \ldots, n_t)$ is upper bounded by $\Phi(n_1, n_2, \ldots, n_t)$. 
%\item For an integer $k$ and a polynomial $f$, $\hom_{k}(f)$ is equal to the homogeneous component of $f$ of degree equal to $k$, while $\hom_{\leq k}(f)$  denotes $\hom_{\leq k}(f) := \sum_{i = 0}^k \hom_i(f)$. We note that $\hom_{\leq k}(f(\vx)) = (f \mod \langle\vx \rangle^{k+1})$. 
\end{itemize}

\subsection{Arithmetic Circuits}
Arithmetic circuits (also historically referred to as straight line programs) provide a succinct and compact representation for multivariate polynomials. Formally, they are defined as follows.
\begin{definition}[Arithmetic Circuit]
Let $\F$ be a field and $\vx = (x_1, x_2, \ldots, x_n)$ be variables. An arithmetic circuit $C$ over $\F$ and $\vx$ is a directed acyclic graph where the vertices are called gates. Every gate with in-degree zero is an input gate and is labeled by a single variable from $\vx$ or a field element from $\F$. The other gates are labeled by $+$ (sum gates) or $\times$ (product gates). The gates with out-degree zero are called output gates.

Each gate in the circuit $C$ computes a polynomial in $\F[\vx]$ in a natural and inductive way. For an input gate $g$, the polynomial it computes is the corresponding variable or field element. A $+$ (resp. $\times$) gate $g$ computes the sum (resp. products) of the polynomials computed at the gates which have a directed edge to $g$. 
%For gate $g$ with label $+$ (resp. $\times$) and incoming gates $g_1,\dots,g_t$, it computes the polynomial $C_g=\sum_{i\in[t]}C_{g_i}$ (resp. $C_g=\prod_{i\in[t]}C_{g_i}$). We say that an arithmetic circuit $C$ computes polynomials $f_1,f_2,\dots,f_t\in\F[\vx]$ if the output gates $g_1,g_2,\dots,g_t$ of $C$ compute these polynomials.
The size of an arithmetic circuit $C$ is defined as the number of edges in $C$. 
\end{definition}
The following lemma structural lemma about arithmetic circuits will be useful for our proof. 
\begin{lemma}[Homogenization]\label{lem:homogenization}
Let $C$ be a multi-output arithmetic circuit of size $s$ with outputs $f_1, f_2, \ldots, f_t$.  Then, for every $k$, there is a homogeneous circuit of size at most $O(k^2s)$ which outputs  the homogeneous components of degree at most $k$ of $f_1, f_2, \ldots, f_t$.
\end{lemma}
We refer the reader to any standard resource (such as the survey by Shpilka and Yehudayoff~\cite{SY10}) for a proof for~\autoref{lem:homogenization} and for a general overview of arithmetic circuit complexity.
\subsection{Multivariate Taylor's Expansion}
We use the following lemma  which is an easy consequence of the classical multivariate Taylor expansion for polynomials. 
\begin{lemma}[Truncated Multivariate Taylor's Expansion]\label{lem:multivariate taylor}
Let $f\in\F[\vx]$ and $\va\in\F^n$, we have
$$
f(\vx + \va) \equiv f(\va) + \sum_{i=1}^m\frac{\partial f}{\partial x_i}(\va)\cdot x_i \mod \langle \vx \rangle^{2} \, . 
$$
\end{lemma}
In the proof for~\autoref{thm:main}, we need a variant of this lemma where the variables are from $\F[\vx]$ instead of $\vx$. We state it as a corollary of~\autoref{lem:multivariate taylor}
\begin{corollary}\label{cor:multivariate taylor polynomials}
Let $f_1,f_2,\dots,f_m,p_1,p_2,\dots,p_m\in\F[\vx]$ where $\deg(p_i)\geq k$ for each $i\in[m]$ and $Q\in\F[z_1,z_2,\dots,z_m]$, we have
$$
Q\left(\mathbf{f}+\mathbf{p}\right) \equiv Q\left(\mathbf{f}\right) + \sum_{i\in[m]}\frac{\partial Q}{\partial z_i}\left(\mathbf{f}\right)\cdot p_i \mod \langle \vx \rangle^{k+1} \, .
$$
\end{corollary}
%\Cnote{Shall we mention that Newton iteration is simple doing Taylor's expansion?}
\subsection{Jacobian Matrix}
The Jacobian matrix is a matrix that contains the partial derivatives of a vector of multivariate functions.
\begin{definition}[Jacobian Matrix]\label{def:jacobian}
Let $f_1,f_2,\dots,f_m\in\F[\vx]$, the Jacobian matrix of $\mathbf{f}$ with respect to $\vx$ is a $m\times n$ matrix denoted as $\mathsf{Jacobian}_\vx(\mathbf{f})$ where the $(i,j)$ entry is defined as $\frac{\partial f_i}{\partial x_j}$ for each $i\in[m]$ and $j\in[n]$.
\end{definition}
%Jacobians are well studied objects in algebraic complexity with various applications in questions on pseudorandomness and algebraic independence. For our application, Jacobians will be used in a fairly intuitive way
%\Cnote{Shall we restate multivariate taylor for vector function?}

\subsection{GCD and Resultant}
For any two polynomials $g,h\in\F[\vx]$, we can define their greatest common divisor (GCD) as follows.
\begin{definition}[GCD]\label{def:gcd}
Let $\F$ be a field and $g,h\in\F[\vx]$. The greatest common divisor (GCD) of $g$ and $h$ is $\mathsf{GCD}(g,h)=f$ if $f$ divides both $g$ and $h$, and for any $\tilde{f}\in \F[\vx]$ that divides both $g$ and $h$, we have $\tilde{f}$ dividing $f$.
For any $g,h\in\F[\vx][y]$, we define $\mathsf{GCD}_y(g,h)$ to be the greatest common divisor $(\mathsf{GCD})$ of $g$ and $h$ with respect to $y$.
\end{definition}

It turns out that there is a clean and useful mathematical condition to check whether the $\mathsf{GCD}$ of two polynomials is non-constant using \textit{resultant}.

\begin{definition}[Resultant]\label{def:resultant}
Let $g,h\in\F[\vx][y]$ and $d_1,d_2\in\N$ such that $g=\sum_{i=0}^{d_1}g_iy^i$ and $h=\sum_{j=0}^{d_2}h_jy^j$ for some $g_i,h_j\in\F[\vx]$. The resultant $R_y(g,h)$ is the determinant of a following $(d_1+d_2) \times (d_1+d_2)$ matrix, called the Sylvester matrix $S(g,h)$.
\[
S(g,h) = \begin{pmatrix}
g_0&0&\cdots&0&h_0&0&\cdots&0\\
g_1&g_0&\cdots&0&h_1&h_0&\cdots&0\\
g_2&g_1&&\vdots&\vdots&h_1&&\vdots\\
\vdots&\vdots&\ddots&g_0&h_{d_2}&\vdots&\ddots&0\\
g_{d_1}&g_{d_1-1}&&g_1&0&h_{d_2}&&h_0\\
0&g_{d_1}&&g_2&0&0&&h_1\\
\vdots&\vdots&&\vdots&\vdots&\vdots&&\vdots\\
0&0&\cdots&g_{d_1}&0&0&\cdots&h_{d_2}
\end{pmatrix} \, .
\]
Specifically, for $i \in \{1, 2, \ldots, d_2\}$, the $i^{th}$ column of $S$ is equal to $(0, \ldots, 0, g_{d_1}, g_{d_1-1}, \ldots, g_1, g_0, 0, \ldots, 0)$, where there are $i-1$ zeroes in the prefix. For $j \in \{d_2+1, \ldots, d_1+d_2\}$, t he $j^{th}$ column of $S$ equals $(0, \ldots, 0, h_{d_2}, h_{d_2-1}, \ldots, h_1, h_0, 0, \ldots, 0)$, where there are $j - d_2-1$ zeroes in the prefix.
\end{definition}
%\Cnote{I defined with non-monic polynomial and the last paragraph can be deleted if it occupy too much space.}
The following lemma shows that $R_y(g,h)=0$ if and only if $g$ and $h$ have a common non-constant factor.
\begin{lemma}[Capturing the GCD via the Resultant]\label{lem:gcd resultant}
Let $g,h\in\F[\vx][y]$ and $d_1,d_2\in\N$ such that $g=\sum_{i=0}^{d_1}g_iy^i$ and $h=\sum_{j=0}^{d_2}h_jy^j$ for some $g_i,h_j\in\F[\vx]$. Then, $R_y(g,h)=0$ if and only if $\mathsf{GCD}_y(g,h)$ has degree at least $1$ in $y$. 
\end{lemma}
To keep this note short, we refer the reader to any standard resource (such as the lecture notes by Sudan~\cite{sud98}) for a proof for~\autoref{lem:gcd resultant}.
%\subsection{Homogenization}
%For an integer $k$ and a polynomial $f$, $\hom_{k}(f)$ is equal to the homogeneous component of $f$ of degree equal to $k$, while $\hom_{\leq k}(f)$  denotes $\hom_{\leq k}(f) := \sum_{i = 0}^k \hom_i(f)$. We note that $\hom_{\leq k}(f(\vx)) = (f \mod \langle\vx \rangle^{k+1})$. 

\section{Proof of Lemma~\ref{clm:distinct factors}}\label{sec:proof claim 1}
We have polynomials $f$, $g$ and $h$ such that $f = g\cdot h$ and $f$ has an arithmetic circuit of size at most $s$. The goal is to show  that $g$ and $h$ have circuits of size at most $\poly(s,n,d)$. Let $d_1, d_2$ be the degrees of $g$ and $h$ respectively and let $d_1 \leq d_2$.If $g$ and $h$ are variable disjoint then we can obtain a circuit for $g$ by just setting the variables in $h$ to random values such that $h$ does not vanish and then scaling by an appropriate field constant. So, we focus on the interesting case when $g$ and $h$ share at least one common variables. Let $y$ be such a variable.  By taking a random (from a large enough grid) $\va \in \F^n$ and replacing $x_i$ by $x_i + a_iy$, we can guarantee that the coefficient of $y^d$ in $f$, $y^{d_1}$ in $g$ and $y^{d_2}$ in $h$ are all non-zero field elements. Without loss of generality, we assume that these constants are all $1$ (or else we scale everything by a constant). In the rest of  this section, we view the identity $f = g\cdot h$ as an identity in $\F[\vx][y]$. Note that at the end of the above transformation, $\mathsf{GCD}(f,g)$ continues to be $1$ when viewing them as univariates in $y$. We know that $f$ has a small circuit, and the goal is to show that $g$ and $h$ have small circuits. 

Let $f_0, f_1, \ldots, f_{d-1}, g_0, g_1, \ldots, g_{d_1-1}, h_0, h_1, \ldots, h_{d_2-1} \in \F[\vx]$ be polynomials such that 
\[
f := y^d + \sum_{i = 0}^{d-1} f_i y^i \, ,
\]
\[
g := y^{d_1} + \sum_{i = 0}^{d_1-1} g_i y^i \, 
\]
and
\[
h := y^{d_2} + \sum_{i = 0}^{d_2-1} h_i y^i \, .
\]
%Let $f,g,h\in\F[\vx,y]$ where $\vx=(x_1,\dots,x_n)$ such that $f=g\cdot h$, $\deg_{y}(g)=d_1, \deg_{y}(h)=d_2$, and $gcd(g,h)=1$. Let $D=d_1+d_2=\deg_{y}(f)$. Write $f,g,h$ in terms of $y$ as follows.
%$$
%f(\vx,y) = \sum_{i=0}^{D}f_i(\vx)y^i,\ g(\vx,y) = \sum_{j=0}^{d_1}g_j(\vx)y^j,\ h(\vx,y) = \sum_{k=0}^{d_2}h_k(\vx)y^k.
%$$
%\Cnote{We should probably assume wlog $d_1\leq d_2$?}
Now, comparing the coefficients of $y^i$ on both sides in the equality $f=g\cdot h$ gives us a system of polynomial equations in $g_0, g_1, \ldots, g_{d_1-1}, h_0, h_1, \ldots, h_{d_2-1}$ as follows.
\begin{align*}
f_0 &= g_0\cdot h_0\\
f_1 &= g_0\cdot h_1 + g_1\cdot h_0\\
\vdots\\
f_i &= \sum_{j=0}^{\min\{i,d_1\}}g_j\cdot h_{i-j}\\
\vdots\\
f_{d-1} &= g_{d_1-1} + h_{d_2-1}.
\end{align*}
Let $\vu = (u_0, u_1, \ldots, u_{d_1-1})$ and $\vw = (w_0, w_1, \ldots, w_{d_2-1})$ be new sets of variables.  For $\ell \in [d]$  define the polynomials 
$$Q_{\ell}(\vu, \vw) := \sum_{j = 0}^{\min\{\ell, d_1-1\}} u_{j}\cdot w_{\ell - j} - f_{\ell}\, .$$
We view $Q_{\ell}$ as a polynomial in $\vu, \vw$ with coefficients coming from the ring $\F[\vx]$. In this sense, $(\vg, \vh) = (g_0, g_1, \ldots, g_{d_1-1}, h_0, h_1, \ldots, h_{d_2-1})$ is a common zero of $Q_0, Q_1, \ldots, Q_{d-1}$. Our goal is to essentially \emph{solve} the system of equations given by $\{Q_{\ell}(\vu, \vv) = 0 : \ell \in [d]\}$ to recover circuits for each $g_i$ and $h_i$ and prove an upper bound on their size. Note that this would not be an efficient  algorithmic procedure, but we will be able to argue about the circuit complexity of the solution. To this end, we first observe some elementary properties of this system of polynomial equations. 
\begin{observation}\label{obs:ckts for Q}
For every $\ell \in [d-1]$, $Q_{\ell}$ can be computed by a circuit of size at most $O(sd)$. 
\end{observation}
\begin{proof}
Since $f$ has  a circuit of size at most $s$, and has degree $d$, each $f_i$ can be computed by a circuit of size at $O(sd)$ by an easy application of~\autoref{lem:homogenization}. This immediately gives a circuit of this size for each $Q_{\ell}$. 
\end{proof}

\begin{lemma}\label{lem:jacobian}
Let ${\cal J}(\vu, \vv) := \mathsf{Jacobian}_{\mathbf{u}, \mathbf{w}}(Q_0, Q_1, \ldots, Q_{d-1})$ be the Jacobian of $Q_0, Q_1, \ldots, Q_{d-1}$. If $\mathsf{GCD}_y(g, h) = 1$, then ${\cal J}(\vg, \vh)$ is a non-singular matrix. 
%, i.e. $J$ is a $d \times d$ matrix, whose rows are labeled by $Q_0, Q_1, \ldots, Q_{d-1}$ and whose columns are labeled by the variables $\mathbf{g}, \mathbf{h}$, and any entry is the derivative of the polynomial labeling the row with respect to the polynomial labeling the column. Then, the determinant of ${\cal J}$ evaluated at $\mathbf{0}$ is non-zero. 
\end{lemma}
\begin{proof}
The key observation here is that the Jacobian matrix (see~\autoref{def:jacobian}) is the same as the \textit{Sylvester matrix} (see~\autoref{def:resultant}) up to a permutation of rows and columns. Concretely, let $R$ be the resultant of $g$ and $h$ when they are viewed as univariates in $y$. Since the $\mathsf{GCD}_y(f,g)$ is equal to $1$, their resultant is a non-zero polynomial of degree at most $O(d^2)$ in $\F[\vx]$. Recall that $R$ is the determinant of the following $d \times d$ matrix, called the Sylvester matrix $S$. For $i \in \{1, 2, \ldots, d_2\}$, the $i^{th}$ column of $S$ is equal to $(0, \ldots, 0, 1, g_{d_1-1}, g_{d_1-2}, \ldots, g_1, g_0, 0, \ldots, 0)$, where there are $i-1$ zeroes in the prefix. For $j \in \{d_2+1, \ldots, d\}$, t he $j^{th}$ column of $S$ equals $(0, \ldots, 0, 1, h_{d_2-1}, h_{d_2-2}, \ldots, h_1, h_0, 0, \ldots, 0)$, where there are $j - d_2-1$ zeroes in the prefix. 
\iffalse
\[
\begin{pmatrix}
g_{d_1}&0&\cdots&0 & h_{d_2} & 0 & \cdots &0\\
g_{d_1-1}&g_{d_1} &\cdots &0 & h_{d_2-1} & h_{d_2} & \cdots &0\\
g_{d_1-2}&g_{d_1-1}&\cdots&0 & h_{d_2-2} & h_{d_2-1} & \cdots &0\\
\vdots&\vdots&\vdots&g_{d_1}&\vdots&\vdots&\vdots&\vdots\\
g_{1}&g_{2}&\cdots&\vdots & h_{d_2-d_1+1} & h_{d_2-d_1+2} & \cdots & \cdot \\
g_{0}&g_{1}&\cdots&\vdots & h_{d_2-d_1} & h_{d_2-d_1+1} & \cdots & \cdot \\
0&g_{0}&\cdots&\vdots & h_{d_1+1} & h_{d_1-1} & \cdots & \cdot \\
\vdots&\vdots&\vdots&\vdots&\vdots&\vdots\\
0&0&\cdots&g_0&0&0&\cdots)h_0
\end{pmatrix}
\]
Here, the first $d_2$ columns have the coefficients of $g$ and the last $d_1$ columns have the coefficients of $h$. For brevity, let us assume that $d_1 < d_2$. 
\fi
We now write the $d\times d$ matrix ${\cal J}(\vu, \vw)$. 
\[
{\cal J}(\vu, \vw)=\begin{pmatrix}\frac{\partial Q_0}{\partial u_0}&\cdots&\frac{\partial Q_0}{\partial u_{d_1-1}}&\frac{\partial Q_0}{\partial w_0}&\cdots&\frac{\partial Q_0}{\partial w_{d_2-1}}\\
\frac{\partial Q_1}{\partial u_0}&\cdots&\frac{\partial Q_1}{\partial u_{d_1-1}}&\frac{\partial Q_1}{\partial w_1}&\cdots&\frac{\partial Q_1}{\partial u_{d_2}}\\
\vdots&\vdots&\vdots&\vdots&\vdots&\vdots\\
\frac{\partial Q_{d-1}}{\partial u_0}&\cdots&\frac{\partial f_{d-1}}{\partial u_{d_1-1}}&\frac{\partial Q_{d-1}}{\partial w_0}&\cdots&\frac{\partial f_{d-1}}{\partial w_{d_2-1}}
\end{pmatrix}
\]
Plugging in the expressions for the partial derivatives, we get that for $i \in \{1, 2, \ldots, d_1\}$, the $i^{th}$ column of ${\cal J}(\vu, \vw)$ is equal to $(0, \ldots, 0, w_0, w_1, w_2, \ldots, w_{d_2-1}, 1, 0, \ldots, 0)$, where there are $i-1$ zeroes in the prefix. For $j \in \{d_1+1, \ldots, d\}$, the $j^{th}$ column of $S$ equals $(0, \ldots, 0, u_0, u_{1}, u_{2}, \ldots, u_{d_1-1}, 1, 0, \ldots, 0)$, where there are $j - d_2-1$ zeroes in the prefix. Therefore, after substitution, the columns of ${\cal J}(\vg, \vh)$ are precisely the same as columns of $S$ up to a permutation of rows and columns. In other words, their ranks are equal. We know that $S$ is non-singular, so it follows that ${\cal J}(\vg, \vh)$ is also non-singular. 
\iffalse

\[
{\cal J}(\vu, \vw)=\begin{pmatrix}w_0&\cdots&0&
u_0&\cdots&0\\
w_1&\cdots&0&u_1&\cdots&0\\
\vdots&\vdots&\vdots&\vdots&\vdots&\vdots\\
\frac{\partial Q_{d-1}}{\partial u_0}&\cdots&\frac{\partial f_{d-1}}{\partial u_{d_1-1}}&\frac{\partial Q_{d-1}}{\partial w_0}&\cdots&\frac{\partial f_{d-1}}{\partial w_{d_2-1}}
\end{pmatrix}
\]

\[
\begin{pmatrix}g_0&\cdots& g_{d_1}& h_0&\cdots&\frac{\partial f_0}{\partial h_{d_2}}\\
\frac{\partial f_1}{\partial g_0}&\cdots&\frac{\partial f_1}{\partial g_{d_1}}&\frac{\partial f_1}{\partial h_1}&\cdots&\frac{\partial f_1}{\partial h_{d_2}}\\
\vdots&\vdots&\vdots&\vdots&\vdots&\vdots\\
\frac{\partial f_{D-1}}{\partial g_0}&\cdots&\frac{\partial f_{D-1}}{\partial g_{d_1}}&\frac{\partial f_{D-1}}{\partial h_0}&\cdots&\frac{\partial f_{D-1}}{\partial h_{d_2}}
\end{pmatrix}
 \]
 \fi
\end{proof}

\begin{remark}\label{rmk:non-singular at origin}
For the rest of the proof, we assume without loss of generality that ${\cal J}(\vg(\mathbf{0}), \vh(\mathbf{0}))$ is non-singular. This follows from the fact that since ${\cal J}(\vg(\mathbf{x}), \vh(\mathbf{x}))$ is non-singular, there is  a $\vb$ such that ${\cal J}(\vg(\mathbf{b}), \vh(\mathbf{b}))$ is non-singular, and up to a translation of the coordinate axes, we can assume that $\vb = \mathbf{0}$. 
\end{remark}

\subsection{Newton Iteration for many variables}
We now show that given the constant term for each polynomial in $\vg, \vh$, we can  recover the polynomials completely. The argument is via a natural and well known multivariate analog of the standard Newton Iteration. Clearly, the constant terms have small circuits (trivial circuits of size $1$), and we show that in this iterative process, we can recover multioutput circuits for $\vg, \vh$ of size $\poly(s, n, d)$. 
\begin{lemma}[One step of Newton Iteration]\label{lem:newton}
 Let $k \geq1$ be any integer.  Let $C_k$ be a multioutput circuit of size most $s_k$ computing polynomials $\tilde{\vg}_k= \left(\tilde{g}_{0,k}, \tilde{g}_{1,k}, \ldots, \tilde{g}_{d_1-1, k}\right), \tilde{\vh}_k = \left(\tilde{h}_{0,k}, \tilde{h}_{1,k}, \ldots, \tilde{h}_{d_2-1, k}\right)$ such that for every $i$ and $j$, 
\[
\tilde{g}_{i,k} \equiv g_i \mod \langle\vx\rangle^k \, , 
\]
and 
\[
\tilde{h}_{j,k}\equiv h_j \mod \langle\vx\rangle^k \, .
\]
Then, there is a constant $c$ independent of $k$ such that the following is true: there is a multioutput circuit $C_{k+1}$ of size at most $s_{k+1} = s_k + (snd)^c$ which computes the polynomials  $\tilde{\vg}_{k+1} = \left(\tilde{g}_{0, k+1}, \tilde{g}_{1, k+1}, \ldots, \tilde{g}_{d_1-1, k+1}\right)$, $\tilde{\vh}_{k+1} = \left(\tilde{h}_{0, k+1}, \tilde{h}_{1, k+1}, \ldots, \tilde{h}_{d_2-1, k+1}\right)$ such that for every $i$ and $j$, 
\[
\tilde{g}_{i, k+1} \equiv g_i \mod \langle\vx\rangle^{k+1} \, , 
\]
and 
\[
\tilde{h}_{j, k+1} \equiv h_j \mod \langle\vx\rangle^{k+1} \, .
\]
%Moreover, $\tilde{\vg}_{k+1}$ and $\tilde{\vh}_{k+1}$ are unique modulo $\langle \vx \rangle^{k+1}$.
\end{lemma}
\begin{proof}
For every $i \in [d_1]$ and $j \in [d_2]$, let $p_i$ and $q_j$ be homogeneous polynomials of degree equal to $k$ such that $\tilde{g}_{i, k}+ p_i \equiv g_i \mod \langle \vx \rangle^{k+1}$ and $\tilde{h}_{i, k} + q_j \equiv h_j \mod \langle \vx \rangle^{k+1}$. Let $\vp$ and $\vq$ be the tuples associated to $p_i's$ and $q_j's$. 
Our goal is to show that these polynomials $\tilde{g}_{i, k}+ p_i$ and $\tilde{h}_{j, k} + q_j$ have \emph{small} circuits. This would complete the proof of the lemma. To this end, we set up a system of linear equations in the $p's$ and $q's$ and show that this system has a unique solution. 

Let $\vg_i$ (resp. $\vh_i$) denote the tuple $(g_0  \mod \langle\vx\rangle^i, g_1 \mod \langle\vx\rangle^i, \ldots, g_{d_1-1} \mod \langle\vx\rangle^i)$ ( resp. $(h_0 \mod \langle\vx\rangle^i, h_1\mod \langle\vx\rangle^i, \ldots, h_{d_2-1} \mod \langle\vx\rangle^i)$). 
For each $\ell\in\{0,1,\dots,d-1\}$, since $Q_{\ell}(\mathbf{g}, \mathbf{h}) = 0$, it follows that 
  $$Q_{\ell}(\mathbf{g}_{k+1}, \mathbf{h}_{k+1}) \equiv 0 \mod \langle \vx \rangle^{k+1} \, . $$ 
By their definition, and the hypothesis of the lemma, we know that $\vp$ and $\vq$ must satisfy 
 $$Q_{\ell}\left(\tilde{\mathbf{g}}_k + \vp, \tilde{\mathbf{h}}_k + \vq\right)\equiv 0 \mod \langle \vx \rangle^{k+1} \, .$$ 
Since each $p_i$ and $q_j$ is a homogeneous polynomial of degree equal to $k$,  via the multivariate Taylor expansion for polynomials (see~\autoref{cor:multivariate taylor polynomials}) for $Q_{\ell}$ around the point $(\tilde{\vg}_k, \tilde{\vh}_k)$ , we  get 
$$Q_{\ell}(\tilde{\mathbf{g}}_k, \tilde{\mathbf{h}}_k)  + \sum_{i\in [d_1]} \frac{\partial Q_{\ell}}{\partial u_i}(\tilde{\mathbf{g}}_k, \tilde{\mathbf{h}}_k)\cdot p_i + \sum_{j\in [d_2]} \frac{\partial Q_{\ell}}{\partial w_j}(\tilde{\mathbf{g}}_k, \tilde{\mathbf{h}}_k)\cdot  q_j\equiv 0 \mod \langle \vx \rangle^{k+1}  \, . $$ 
Note that we used the fact that $p_i$ and $q_j$ are homogeneous polynomials of degree equal to $k$ and hence their squares and higher powers vanish modulo $\langle \vx \rangle^{k+1}$. Moreover,  the only monomials of degree at most $k$ in $\frac{\partial Q_{\ell}}{\partial u_i}(\tilde{\mathbf{g}}_k, \tilde{\mathbf{h}}_k)\cdot p_i$ are those in $\left(\frac{\partial Q_{\ell}}{\partial u_i}(\tilde{\mathbf{g}}_k, \tilde{\mathbf{h}}_k) \mod \langle \vx \rangle\right)\cdot  p_i\, .$
We also know from the hypothesis that $\left(\frac{\partial Q_{\ell}}{\partial u_i}(\tilde{\mathbf{g}}_k, \tilde{\mathbf{h}}_k) \mod \langle \vx \rangle\right)$ is equal to $\left(\frac{\partial Q_{\ell}}{\partial u_i}(\mathbf{g}(\mathbf{0}), \mathbf{h}(\mathbf{0})\right)$.
Applying these simplifications to the Taylor expansion for $Q_{\ell}$, we get 
\begin{eqnarray*}
-Q_{\ell}(\tilde{\mathbf{g}}_k, \tilde{\mathbf{h}}_k)  &\equiv& \sum_{i\in [d_1]} \left(\frac{\partial Q_{\ell}}{\partial u_i}(\mathbf{g}(\mathbf{0}), \mathbf{h}(\mathbf{0}))\right)\cdot p_i 
+ \sum_{j\in [d_2]} \left(\frac{\partial Q_{\ell}}{\partial w_i} (\mathbf{g}(\mathbf{0}), \mathbf{h}(\mathbf{0}))\right)\cdot  q_j  \mod \langle \vx \rangle^{k+1} \, . 
\end{eqnarray*}
Let $\vv = (v_0, v_1, \ldots, v_{d_1-1})$ and $\vz = (z_0, z_1, \ldots, z_{d_2-1})$ be new sets of variables.  For various values of $\ell$, let us consider the affine constraint on these variables given by the equation. 
\begin{eqnarray*}
-Q_{\ell}(\tilde{\mathbf{g}}_k, \tilde{\mathbf{h}}_k) \mod \langle \vx \rangle^{k+1} &=& \sum_{i\in [d_1]} \left(\frac{\partial Q_{\ell}}{\partial u_i}(\mathbf{g}(\mathbf{0}), \mathbf{h}(\mathbf{0}))\right)\cdot v_i + \sum_{j\in [d_2]} \left(\frac{\partial Q_{\ell}}{\partial w_i} (\mathbf{g}(\mathbf{0}), \mathbf{h}(\mathbf{0}))\right)\cdot  z_j  \, . 
\end{eqnarray*}
So, we get a system of non-homogeneous linear equations of the form $A\cdot (\vv, \vz)^{T} - \va = 0$, where the matrix $A$ equals the matrix ${\cal J}(\vg(\mathbf{0}), \vh(\mathbf{0}))$, which by~\autoref{lem:jacobian} is non-singular. Thus, this system has unique solution which is given by $\mathbf{\vz} = A^{-1} \mathbf{a}$. From our set up above, $(\vp, \vq)$ is a solution to this system of equations, and thus by uniqueness of solution, we get that that there are field constants $\set{\beta_{i, \ell} : i \in [d_1], \ell \in [d]}$ and $\set{\gamma_{j, \ell} : j\in [d_2] , \ell\in [d]}$ in $\F$ such that  for every $i\in [d_1]$ and $j\in [d_2]$, 
$$p_i = \sum_{\ell\in [d]} \beta_{i, \ell}\left(Q_{\ell}(\tilde{\mathbf{g}}_k, \tilde{\mathbf{h}}_k) \mod \langle \vx \rangle^{k+1} \right) \, ,$$ 
$$q_j = \sum_{\ell\in [d]} \gamma_{j, \ell}\left(Q_{\ell}(\tilde{\mathbf{g}}_k, \tilde{\mathbf{h}}_k) \mod \langle \vx \rangle^{k+1} \right) \, .$$ 
%where $\beta_{i,\ell}$ are in $\F$. 
In other words, 
$$p_i = \left(\sum_{\ell\in [d]} \beta_{i, \ell}\left(Q_{\ell}(\tilde{\mathbf{g}}_k, \tilde{\mathbf{h}}_k) \right)\right) \mod \langle \vx \rangle^{k+1}\, ,$$ 
and 
$$q_j = \left(\sum_{\ell\in [d]} \gamma_{j, \ell}\left(Q_{\ell}(\tilde{\mathbf{g}}_k, \tilde{\mathbf{h}}_k) \right)\right) \mod \langle \vx \rangle^{k+1}\, ,$$ 
Now, recall that for every $\ell$, $\Q_{\ell}(\tilde{\vg}_k, \tilde{\vh}_k)$ is a polynomial which is zero modulo $\langle \vx \rangle^{k}$. Thus, the polynomial $\tilde{g}_{i,k} + \left(\sum_{\ell} \beta_{i, \ell}\left(Q_{\ell}(\tilde{\mathbf{g}}_k, \tilde{\mathbf{h}}_k) \right)\right)$ is equal to $g_i$ modulo $\langle \vx \rangle^{k+1}$, $\tilde{g}_{i,k}$ agrees with $g_i$ at monomials of degree less than $k$, and we are adding to it the \emph{correct} homogeneous polynomial of degree equal $k$. Thus, we define $\tilde{g}_{i, k+1}$ and $\tilde{h}_{j, k+1}$ as 
\[
\tilde{g}_{i, k+1} := \tilde{g}_{i,k} + \left(\sum_{\ell \in [d]} \beta_{i, \ell}\left(Q_{\ell}(\tilde{\mathbf{g}}_k, \tilde{\mathbf{h}}_k) \right)\right) \, ,
\]
and 
\[
\tilde{h}_{j, k+1} := \tilde{h}_{j,k}+ \left(\sum_{\ell \in [d]} \gamma_{j, \ell}\left(Q_{\ell}(\tilde{\mathbf{g}}_k, \tilde{\mathbf{h}}_k) \right)\right) \, .
\]
All that remains now is to argue that there is a small circuit computing $\tilde{\vg}_{k+1}$ and $\tilde{\vh}_{k+1}$. We can obtain a circuit for computing $\tilde{\mathbf{g}}_{k+1}$ and $\tilde{\mathbf{h}}_{k+1}$ from the circuit computing $\tilde{\mathbf{g}}_{k}$ and $\tilde{\mathbf{h}}_{k}$ by adding a copy of circuits for $Q_0, Q_1, \ldots, Q_{d-1}$ at the top and a layer of addition gates above it with appropriate edge weights. The size therefore increases additively by a fixed polynomial in $s, n, d$ in each step.
%The uniqueness part in the claim follows from the fact that the system of linear equations we solve here has a unique solution. This implies that each solution to this system of equations modulo $\langle \vx \rangle^{k}$ has a unique lift to a solution modulo $\langle \vx \rangle^{k+1}$.  Thus, for any two solutions which agree modulo $\langle \vx \rangle$ must be equal to each other. This completes the proof of the claim. 
\end{proof}
We now complete the proof of~\autoref{clm:distinct factors}.
\begin{proof}[Proof of~\autoref{clm:distinct factors}]
Observe that modulo $\langle \vx \rangle$, the $g_i$ and $h_j$ trivially have a circuit of size $1$, since they are just constants. Now, using this as the base case, we use~\autoref{lem:newton} $d+1$ for times to obtain a multioutput circuit $C_{d+1}$ of size at most $\poly(s, n, d)$, which computes polynomials $\tilde{\vg}_{d+1}$ and $\tilde{\vh}_{d+1}$ such that for every $i$ 
\[
\tilde{g}_{i, d+1} \equiv g_i \mod \langle \vx \rangle^{d+1} \, .
\] 
But, each $g_i$ is of degree at most $d$. Thus, we get 
\[
\tilde{g}_{i, d+1} \mod \langle \vx \rangle^{d+1} = g_i  \, .
\] 
So, we recover a circuit for each $g_i$ and each $h_j$, we  homogenize (see~\autoref{lem:homogenization}) the circuit $C_{d+1}$ to get a homogeneous circuit $C_{d+1}'$, which has an output gate for each homogeneous component of degree at most $d$ for every output of $C_{d+1}$.  This incurs an additional multiplicative blow up of $O(d^2)$ on the size of $C_{d+1}$. From the circuit $C_{d+1}'$, we can just read off the polynomials $g_i$ (resp. $h_j$) by taking an appropriate linear combinations of the outputs of $C_{d+1}'$, which only incurs an additive $\poly(s, n, d)$ blow up in the size of the circuit. This completes the proof of the lemma.
\end{proof}

%\fi
\section{Proof of Lemma~\ref{clm:purepower}}\label{sec:proof claim 2}
We have polynomials $f$ and $g$ such that $f=g^e$ where $f$ has a circuit of size at most $s$. Since $d$ be the degree of $f$,  both $e$ and the degree of $g$ are upper bounded by $d$. The goal is showing that $g$ has a circuit of size at most $\poly(s,n,d)$. Now, consider the following polynomial
\[
\tilde{f} := z^e-f = z^e-g^e
\]
where $z$ is a new variable. Note that $\tilde{f}$ has a circuit of size at most $s+O(\log d)$. Next, decompose $\tilde{f}$ as follows.
\[
\tilde{f} = \left(z-g\right)\cdot\left(z^{e-1}+z^{e-2}g+\cdots+g^{e-1}\right) \, .
\]
Observe that if $g\not\equiv0$ and the characteristic of the field is zero or large enough, then the $\mathsf{GCD}_z$ of the polynomial $\left(z-g\right)$ and the polynomial $\left(z^{e-1}+z^{e-2}g+\cdots+g^{e-1}\right)$ is $1$. The reason is that $\left(z-g\right)$ is irreducible and does not divide $\left(z^{e-1}+z^{e-2}g+\cdots+g^{e-1}\right)$ when $g\not\equiv0$ and the characteristic of the field is zero or large enough. Finally, by~\autoref{clm:distinct factors}, $\left(z-g\right)$ has a circuit of size at most $\poly(s,n,d)$ and thus $g$ also has a circuit of size $\poly(s, n, d)$.

\section*{Acknowledgements}
We thank Vishwas Bhargav, Swastik Kopparty, Ramprasad Saptharishi and Srikanth Srinivasan for various insightful discussions and for encouraging us to write the proof up. 
\bibliographystyle{alphaurlpp}
\bibliography{references}

\end{document}